\documentclass[11pt,a4paper,english]{article}

\title{Remarks on QFT in the Coordinate Space}
\author{A. Much\\ \footnotesize{Instituto de Ciencias Nucleares, UNAM, M\'exico D.F. 04510, M\'exico},
	\\ \footnotesize{Faculty of Mathematics, University of Vienna, 1090 Vienna, Austria}}

\usepackage{float}            % bietet Option [H] fÃƒÂ¼r bombenfestes Verankern
\usepackage[T1]{fontenc}   
\usepackage[cmex10]{amsmath}  % fÃƒÂ¼r erw. Formeloptionen, Option [] zur Vermeidung von Type3-Fonts
\usepackage{amstext}          % fÃƒÂ¼r Klartext via \text{} in Formeln
\usepackage{mathrsfs}
\usepackage{amsfonts}         % fÃƒÂ¼r komplexere Formeln (Mengensymbole ...)
\usepackage{amssymb}          % fÃƒÂ¼r komplexere Formeln (Mengensymbole ...)
\usepackage{amsthm}    
\usepackage{setspace}
\usepackage{bm}               % bold math, fÃƒÂ¼r \bm{}
\usepackage{enumerate}        % verbessert AufzÃƒÂ¤hlungen
\usepackage[bottom]{footmisc} % Fussnoten am Seitenende
\usepackage{array}            % fÃƒÂ¼r Tabellen: bindet tabular-Umgebung ein
\usepackage{textcomp}         % fÃƒÂ¼r \textdegree , \textcelsius , macht aber manchmal Probleme (!!)
\usepackage{pdfpages}         % fÃƒÂ¼r die Einbindung kompletter pdf-*Seiten*
\usepackage{parskip}          % zw. AbsÃƒÂ¤tzen: eine knappe Leerzeile statt hÃƒÂ¤ngender EinzÃƒÂ¼ge
\usepackage[right]{eurosym}   % Eurosymbol
\usepackage{xcolor}           % farbiger Text
\usepackage[square,numbers]{natbib}	  % FÃƒÂ¼r \setlength{\bibsep}{3mm}; square macht eckige

\usepackage{amsmath}  
\usepackage[colorlinks,pdfpagelabels,pdfstartview = FitH,bookmarksopen = true,bookmarksnumbered =
true,linkcolor = black,plainpages = false,hypertexnames = false,citecolor = black]{hyperref}

\newtheorem{theorem}{\textsc{Theorem}}[section]
\newtheorem{lemma}{\textsc{Lemma}}[section]
\newtheorem{proposition}{\textsc{Proposition}}[section]

\theoremstyle{definition}
\newtheorem{definition}{\textsc{Definition}}[section]

\newtheorem{convention}{Conventions}[section]

\theoremstyle{remark}
\newtheorem{remark}{Remark}[section]

\newcommand{\R}{\mathbb{R}}

\usepackage{fancyhdr}
\numberwithin{equation}{section} 
\begin{document}
	\maketitle
	\abstract{Generators of the Poincar\'e  group, for a free massive scalar field, are usually expressed  in  the momentum space.  In this  work we perform a transformation of these generators	into the coordinate space. This (spatial)-position space is spanned by  eigenvectors of the Newton-Wigner-Pryce operator. The motivation  is twofold. First, we want to investigate the  localization of a relativistic  particle. Furthermore, we need a deeper understanding of the commutative spatial coordinate space in QFT, in order to investigate the non-commutative version thereof.} 
	%%%%%%%%%%%%%%%%%%%%%%
	% Inhaltsverzeichnis %
	\tableofcontents
	%%%%%%%%%%%%%%%%%%%%%%
	
	\section{Introduction}
	Form an algebraic point of view,  standard problems in  quantum mechanics can be  described in two spaces,
	the momentum and the coordinate space. These  spaces are generated  by two complimentary observables,
	which are each represented by    self-adjoint operators. These  operators  are the momentum and the coordinate operator. In particular, the  spaces they generate  consist of their respective eigenvectors. Both spaces are equally relevant for investigation and the choice usually depends on the nature of the problem. 
	\\\\Nevertheless, due to the  complementarity,  it  is   difficult, if not impossible,  to understand quantum mechanics from both perspectives at the  \textbf{same time}. In a letter from   Pauli to Heisenberg (19. October 1926), there is the famous quote describing this particular  problem:
	\textit{"One can see the world with the p-eye or one can
		see the world with the x-eye, if one opens both eyes
		at once, one becomes crazy." }
	\\\\ However, in quantum field theory (QFT) on Minkowski space  the world   is mostly seen  through   the eye of  momentum space. This is partly due to the Wigner classification in terms of   unitary,   irreducible representations   of the Poincar\'e group.  Hence, most work is done in the momentum basis and the conjugate basis of the momentum, i.e. the spatial coordinate space is, for the most part, overlooked. \\\\
	It should be noted that  there exists an appropriate operator,  whose eigenstates span the spatial position space. It is in literature mostly referred to as the Newton-Wigner-Pryce (NW)-operator, see  \cite{PR48}, \cite{NW49},  \cite{Sch}, \cite{J80}, \cite{SS} and \cite{Muc3}. The importance of these states relies on the fact that they can be used to calculate the probability amplitude  of finding a particle at a certain spatial-position at time $t$. Moreover, the second quantization of this operator allows us  to calculate the probability amplitude  of finding $k$-particles at   certain spatial-positions at time $t$.
	\\\\The  motivation  to perform investigations  in the coordinate space, comes from the desire to understand and investigate the non-commutative version thereof.  Hence, concerning investigations in QFT that assume  a constant non-commutative space (without the involvement of time) the coordinate space is of great importance. In particular, it is the space where non-commutativity is introduced. Therefore, in order to understand how the  Poincar\'e group  acts or which form the generators take on   non-commutative spaces, in a quantum field theoretical context, we first need to investigate the commutative framework.
	\\\\
	Therefore, in this work we open the other eye, i.e. take a  first step in this direction by expressing  second-quantized quantities of the  Poincar\'e algebra in the coordinate basis. In particular, we express  generators of the respective algebras in terms of Fock space operators, which are usually written in the momentum space, and change the basis  to the coordinate space. This is done by using the eigenvectors and corresponding ladder operators of the spatial coordinate operator. Moreover, we investigate how the coordinate space behaves under transformations of the Poincar\'e group. 
	\\\\ 
	It is important to point out that there is a clear distinction between the so-called configuration space (see \cite[Chapter 7]{Sch}) and the complimentary space of the momentum, i.e. the coordinate space. In QM these two spaces are equal, however in the QFT-case these two cases are different, but related by an integral transformation. 
	\\\\
	The paper is organized as follows; Section two comprises the  preliminaries, where we define the Fock space of the free scalar field and  the algebra of the Fourier transformed ladder operators. The third section gives a treatment of the  Newton-Wigner-Pryce operator, i.e. the respective spatial-position coordinate operator.  In section four we transform the Poincar\'e algebra of a massive scalar field into the coordinate space. 
	Section five investigates the transformation behavior of the coordinate space under the Poincar\'e group. 
	\begin{convention}
		We use $d=n+1$, for $n\in\mathbb{N}$ and the Greek letters are split into  $\mu, \,\nu=0,\dots,n$. Moreover, we use Latin letters for the spatial components which run from $1,\dots,n$ and we choose the following convention for the Minkowski scalar product of $d $-dimensional vectors, $a\cdot b=a_0b^0+a_kb^k=a_0b^0- \vec{a}\cdot\vec{b}$.
	\end{convention}
	\newpage
	\section{Preliminaries}

	\subsection{Bosonic Fock Space and Fourier-Transformation} The ($n$+$1$)-dimensional ($n$$\in$$\mathbb{N}$) relativistic  Bosonic Fock 
	space is defined in the following. Let a particle have momentum $\mathbf{p} \in \mathbb{R}^n$. Then  the energy of a massive  particle is  $\omega_{\mathbf{p}}=+\sqrt{\mathbf{p}^2+m^2}$. In addition, the Lorentz-invariant measure is given by   $d^n\mu(\mathbf{p} )=d^n\mathbf{p}/( {2\omega_{\mathbf{p}}})$.
	\begin{definition}\label{bf}
		The \textbf{Bosonic Fock space} $\mathscr{F^{+}({H})}$ is defined 
		as in \cite{BR}:
		\begin{equation*}
		\mathscr{F^{+}({H})}=\bigoplus_{k=0}^{\infty}\mathscr{H}_{k}^{+},
		\end{equation*}
		where $\mathscr{H}_{0}=\mathbb{C}$ and  the symmetric $k$-particle subspaces are given as
		\begin{align*}
		\mathscr{H}_{k}^{+}&=\{\Psi_{k}: \underbrace{  \mathscr{H}_{1}  \times  \dots \times 
			\mathscr{H}_{1}}_{k-times} \rightarrow \mathbb{C}\quad \mathrm{symmetric}
		|\\ &\left\Vert  \Psi_k \right\Vert^2 =\int 
		d^n\mu(\mathbf{p}_1)\dots\int d^n\mu(\mathbf{p}_k)
		|\Psi_{k}(\mathbf{p}_1,\dots,\mathbf{p}_k)|^2<\infty\},
		\end{align*}
		with $\mathscr{H}_{1}$ being defined by 
		\begin{align*}	\mathscr{H}_{1} &=\{\Psi :   	H_{m}^{+} \rightarrow \mathbb{C} 
		| \left\Vert  \Psi  \right\Vert^2 =\int 
		d^n\mu(\mathbf{p} ) 
		|\Psi (\mathbf{p} )|^2<\infty\},
		\end{align*}
		where $	H_{m}^{+}$ is the orbit 
		\begin{align*}
		H_{m}^{+}&: =\{p\in \mathbb{R}^{n+1}|p^2=m^2,p_0>0\}.
		\end{align*} 
	\end{definition}$\,$\\
The ladder operators $a,a^{*}$ for  this particular space are defined as follows. \newline $\,$
	\begin{definition} 
		The covariant \textbf{particle annihilation} and \textbf{creation operators} are defined by their 
		action on $k$-particle wave functions,
		\begin{align*}
		(a_c(f)\Psi)_k(\mathbf{p}_1,\dots,\mathbf{p}_k)&=\sqrt{k+1}\int 
		d^n\mu(\mathbf{p})\overline{f(\mathbf{p})}
		\Psi_{k+1}(\mathbf{p},\mathbf{p}_1,\dots,\mathbf{p}_k)\\
		(a_c(f)^{*}\Psi)_k( \mathbf{p },\mathbf{p}_1,\dots,\mathbf{p}_k)&= \left\{
		\begin{array} {cc}
		0, \qquad &k=0 \\ \frac{1}{\sqrt{k}}\sum\limits_{i=1}^{k} f(\mathbf{p}_i)
		\Psi_{k-1}(\mathbf{p}_1,\dots,\mathbf{p}_{i-1},\mathbf{p}_{i+1},\dots,\mathbf{p}_k),\quad 
		&k>0
		\end{array} \right.
		\end{align*}
		with $f \in \mathscr{H}_{1} $ and $\Psi_k \in \mathscr{H}_{k}^{+}$ . The commutator
		relations of $a_c(f), 
		a_c(f)^{*}$ follow immediately and are given as  
		\begin{align*}
		[a_c(f), a_c(g)^{*}]=(f,g)=\int d^n\mu(\mathbf{p}) \overline{f(\mathbf{p})} 
		g(\mathbf{p}), \qquad
		[a_c(f), a_c(g)]=0=[a_c(f)^{*}, a_c(g)^{*}].
		\end{align*}
		Particle annihilation and creation operators with sharp momentum are 
		introduced as operator valued distributions and are given by
		\begin{align*}
		a_c(f)=\int d^n\mu(\mathbf{p}) \overline{f(\mathbf{p})}a_c(\mathbf{p}), 
		\qquad a_c(f)^{*}=\int d^n\mu(\mathbf{p}) {f(\mathbf{p})}a_c^{*}(\mathbf{p}),
		\end{align*}
		where the particle annihilation and creation operators with sharp 
		momentum satisfy the following commutator relations
		\begin{align}\label{pccr}
		[a_c(\mathbf{p}), a_c^{*}(\mathbf{q})]=2\omega_{\mathbf{p} 
		}\delta^n(\mathbf{p}-\mathbf{q}), \qquad
		[a_c(\mathbf{p}), a_c(\mathbf{q})]=0=[a_c^{*}(\mathbf{p}), a_c^{*}(\mathbf{q})].
		\end{align}
	\end{definition}
	From a book keeping point of view, the noncovariant  representation of the  annihilation and creation operators given as,
	\begin{equation*} 
	{a}(\textbf{p}):= \frac{{a}_c (\textbf{p})}{\sqrt{2\omega_{\mathbf{p}}}},\qquad  {a}^* (\textbf{p}):= \frac{{a}_c^* (\textbf{p})}{\sqrt{2\omega_{\mathbf{p}}}},
	\end{equation*}
	is easier to handle. Hence, in the following sections, when explicit calculations are performed, we     use the non-covariant representation. In order to give the base change of  the infinitesimal generators of the   Poincar\'e  group, in terms of the Fourier-transformed creation and annihilation operators, we first explicitly define the base change.
	\begin{definition}\label{sec2.2}\textbf{Fouriertransformation}\\\\
		The base change to the Fourier-transformed creation and annihilation operators is given by,
		\begin{align*}
		{a}(\mathbf{p})=(2\pi)^{-n/2} \int
		d^n \mathbf{x}\, e^{ip_kx^k} \tilde{a}(\mathbf{x}),\qquad {a}^{*}(\mathbf{p})=(2\pi)^{-n/2} \int
		d^n \mathbf{x}\, e^{-ip_kx^k} \tilde{a}^*(\mathbf{x}).
		\end{align*}
		
		From the commutation relation between the two operators $a$ and $a^{*}$ we can deduce   the relations for the Fourier-transformed operators, 
		\begin{align*}
		\delta^n(\mathbf{p}-\mathbf{q})=[{a}(\mathbf{p}),{a}^*(\mathbf{q})]=(2\pi)^{-n} 
		\iint
		d^n \mathbf{x}\,d^n \mathbf{y}\, e^{ip_kx^k}\, e^{-iq_ky^k} [\tilde{a}(\mathbf{x}),\tilde{a}^*(\mathbf{y})].
		\end{align*}
		Hence the solution of the commutation relations of the coordinate space creation and annihilation operators is given, as follows
		\begin{align*}
		[\tilde{a}(\mathbf{x}),\tilde{a}^*(\mathbf{y})]=\delta^n(\mathbf{x}-\mathbf{y}).
		\end{align*}
		The inverse transformations are given by,
		\begin{align}\label{inf}
		\tilde{a}(\mathbf{x})=(2\pi)^{-n/2} \int
		d^n \mathbf{p}\, e^{-ip_kx^k}{a}(\mathbf{p}),\qquad 
		\tilde{a}^*(\mathbf{x})=(2\pi)^{-n/2} \int
		d^n \mathbf{p}\, e^{ ip_kx^k}{a}^*(\mathbf{p}).
		\end{align}

	\end{definition}
	This commutator could also been have obtained by taking the CR's (commutator relations) of the smeared operators into account. Moreover, on   eigenvectors of the coordinate operator, the Fourier-transformed creation and annihilation operators act $\tilde{a},\,\tilde{a}^*$ as ladder operators. Note that the definition of the Fourier-transformed operators does not explicitly  depend on the mass.

	\subsection{Constantly Used Integrals}
	In order to make this work self-contained we give the general formula for certain Fourier-transformed functions,  \cite[Chapter III, Section  2.8]{GS1}
	\begin{align}\label{f2}
	\tilde{P}^{\lambda}&= \int\,d^n\mathbf{p}\, \left(\vert\vec{p}\vert^2+m^2\right)^{\lambda} \exp({-i\,p_k z^k})\\\nonumber
	&= \frac{2^{ \lambda +1}(2\pi)^{\frac{1}{2}n}}{\Gamma(-\lambda)} \left(\frac{m}{\vert \vec{z}\vert}\right)^{\frac{1}{2}n+\lambda}K_{\frac{1}{2}n+\lambda}(m\vert \vec{z}\vert ),
	\end{align}
	where $K$ denotes the modified Bessel-functions of second order and $\Gamma$ is the Gamma function. 
	
	\newpage
	\section{Newton-Wigner-Pryce Operator}\label{NWPsec}
	As already pointed out in the introduction, the appropriate (spatial)-position operator for the Klein-Gordon field is given by the so-called Newton-Wigner-Pryce  operator, \cite{PR48} and \cite{NW49}. For the one-particle case it is given  by the following action, \cite[Chapter 3c, Equation 35]{Sch}
	\begin{equation}\label{NWP}
	(X_{j} \varphi)(\mathbf{p})=-i \left( \frac{p_j}{2\omega_{\mathbf{p}}^2}
	+   \frac{\partial}{\partial p^j } 
	\right)\varphi(\mathbf{p}),
	\end{equation}
	where here we use the covariant representation and normalization, i.e. the observable $X_{j}$ is  represented as a self-adjoint operator  w.r.t. the $\mathscr{H}_1$-scalar product. For the free (relativistic) scalar field the eigenfunctions of the Newton-Wigner-Pryce  operator, which are simultaneously the localized wave functions at time $x_{0}=0$, are given by, \cite[Chapter 3, Equation 38]{Sch}
	\begin{align*} 
	\Psi_{\mathbf{x},0}(\mathbf{p})=(2\pi)^{-n/2}\,e^{-i\mathbf{p} \cdot \mathbf{x}}\,(2\omega_{p})^{1/2}.
	\end{align*}  
	They are of physical importance from the following point of view. Let a particle be in a state $\Phi(\mathbf{p})$ at time $t=0$. Then, the probability amplitude that a position measurement will find the particle at $\mathbf{x}$ is given by, \cite[Chapter 3, Equation 44]{Sch}
	\begin{align*} 
	\langle \overline{\Psi}_{\mathbf{x}}, \Phi\rangle=\int d^n\mu(\mathbf{p})\, \Psi_{\mathbf{x},0}(\mathbf{p}) \Phi(\mathbf{p}).
	\end{align*} 
	In \cite{SS}, \cite{Muc2} and \cite{Muc3} a second-quantized version of this operator was given.  Hence, the $k$-particle generalization of the Newton-Wigner-Pryce operator   can be defined.  Particularly, it can be obtained by the Fourier-transformation of the spatial momentum operator. It is given as a   self-adjoint operator on the domain $\bigotimes_{i=1}^k \mathscr{S}(\mathbb{R}^n)$, with  $\mathscr{S}(\mathbb{R}^n)$ denoting the Schwartz space, as follows, \cite{Muc2, SS}
	\begin{align}\label{cop}
	X_j&=-i\int d^n \mathbf{p}\,  {a}^*(\textbf{p}) \frac{\partial}{\partial p^j}  {a}(\textbf{p})
	,
	\end{align}
	where the  second-quantized  momentum operator, is given as a self-adjoint operator on the domain $\bigotimes_{i=1}^k \mathscr{S}(\mathbb{R}^n)$ and  in  terms of the ladder operators   as
	\begin{align*}
	P_{\mu}&= \int d^n \mathbf{p}\, p_{\mu}\, {a}^*(\mathbf{p}) {a}(\mathbf{p})
	.
	\end{align*}The commutator of  the spatial-momentum operator and the NWP-operator is simply the second-quantized Heisenberg-Weyl relation, see \cite{SS} or \cite{Muc3},
	\begin{align}\label{ccr}
	[X_j,P_k]&=-i\eta_{jk}N,
	\end{align} 
	where $N$ is the particle-number operator represented in Fock-space as
	\begin{align}\label{pn}	N=  \int
	d^n \mathbf{p}\,     {a}^*(\mathbf{p}) {a}(\mathbf{p}).
	\end{align}
	In what follows we give  the form and proof of the coordinate operator in coordinate space.  
	\begin{lemma}\label{nwp}
		The  Newton-Wigner-Pryce operator has the following coordinate space representation,
		\begin{align*}
		X_j&= \int d^n \mathbf{x}\,  x_j \, \tilde{a}^*(\textbf{x}) \tilde{a}(\textbf{x})
		.
		\end{align*}
	\end{lemma}
	\begin{proof}
		The proof is done by changing from the momentum basis to the coordinate basis, i.e.
		\begin{align*}
		X_j&=-i\int d^n \mathbf{p}\,  {a}^*(\textbf{p}) \frac{\partial}{\partial p^j}  {a}(\textbf{p})\\&=-i
		(2\pi)^{-n}\int d^n \mathbf{p}\, \int
		d^n \mathbf{x}\, e^{-ip_rx^r} \tilde{a}^*(\mathbf{x})
		\int
		d^n \mathbf{y}\, \left(\frac{\partial}{\partial p^j}e^{ip_ky^k} \right)\tilde{a} (\mathbf{y})
		\\&= 
		(2\pi)^{-n}\,\iint
		\, d^n \mathbf{x}\,  d^n \mathbf{y} \, y_j \underbrace{\left( \int d^n\mathbf{p}\, e^{-ip_k(x-y)^k}\right)}_{(2\pi)^{n}\delta(\mathbf{x}-\mathbf{y})} \tilde{a}^*(\mathbf{x})
		\tilde{a} (\mathbf{y}) ,
		\end{align*}
		where in the last lines we performed the derivative and integrated over the delta distribution. 
	\end{proof}  
	Next, we want to compare the second-quantized operator with the $k$-particle extension in literature and     draw attention to the distinction between   coordinate space and       configuration space (see \cite[Chapter 7]{Sch}). In order to do so let us introduce the configuration space in QFT. For a free scalar field a one particle state is given    by
	\begin{align*}\vert x\rangle=
	\phi^{-}(x)\vert0\rangle &= \int d^n\mu(\mathbf{p}) \,e^{ipx}a_{c}^{*}(\mathbf{p})\vert0\rangle  
	= \int d^n\mu(\mathbf{p}) \,e^{ipx}  \vert \mathbf{p}\rangle,
	\end{align*}  
	where $\phi^{-}(x)$ contains only negative-frequency parts of the Klein-Gordon field at a \textbf{space-time} point $x$. One can think of the configuration  space as the space spanned by vectors, 
	\begin{align*}\vert x_{1},\cdots,x_{k}\rangle=(k!)^{-1/2}
	\phi^{-}(x_{1})\cdots\phi^{-}(x_{k})\vert 0\rangle, 
	\end{align*}
	where the vector $\vert 0\rangle$ is the vacuum. The $k$-particle Fock-space amplitude for a vector $\vert\Psi\rangle$ is given by the scalar product with the configurations space vectors, i.e. 
	\begin{align*}
	\Psi^{(k)}(x_{1},\cdots, x_{k})=(k!)^{-1/2} \langle 0\vert \phi^{+}(x_{1})\cdots\phi^{+}(x_{k}) 
	\vert\Psi  \rangle,
	\end{align*}
	where $\phi^{+}(x)$ contains only positive-frequency parts of the Klein-Gordon field at a  point $x$. From a physical point of view, this is \textbf{not} the probability amplitude for finding   $k$-particle at position $\mathbf{x}_{1} \cdots \mathbf{x}_{k}$ at time $x_{0}={x}_{10}= \cdots{x}_{k0}$. But it is rather the probability amplitude for finding   $k$-particle at time $x_{0}$. Hence, in order to find the quantity,  which gives as the probability amplitude for finding   $k$-particle at positions $\mathbf{x}_{1} \cdots \mathbf{x}_{k}$ at time $x_{0}={x}_{10} \cdots{x}_{k0}$, as we did before for one-particle, we introduce the following operator,  \cite[Chapter 7, Equation 99]{Sch},
	\begin{align}\label{copes}
	\phi_1( {x})&=\int d^n\mu(\mathbf{p}) \overline{\Psi_{\mathbf{x},x_{0}}(\mathbf{p})} a_c(\mathbf{p})\\&=(2\pi)^{-n/2}\,
	\int d^n\mu(\mathbf{p}) e^{i {p} \cdot {x}}\,(2\omega_{p})^{1/2} a_c(\mathbf{p}),
	\end{align}
	and apply $k$ of them on $\vert\Phi^1  \rangle$ and the vacuum as  follows,
	\begin{align*} 
	\Phi^1((x_{0},{x}_1),\cdots,(x_{0},\mathbf{x}_k)) = (k!)^{-1/2} \langle 0\vert \phi_1(x_{0},\mathbf{x}_{1})\cdots\phi_1(x_{0},\mathbf{x}_{k})
	\vert\Phi^1  \rangle.
	\end{align*}
	Although we naively introduced in the preliminaries the Fourier transformation of the creation and annihilation operators, we have the following equality 
	\begin{align}\label{copes1}
	\phi_1({x})\vert_{x_0=0}=\tilde{a}(\mathbf{x}).
	\end{align}
	The  equivalence is easily proven by taking the non-covariant normalization  and the inverse Fourier-transformation into account. The eigenstates of the second-quantized coordinate operator are created by the action of $\tilde{a}(\mathbf{x}_1),\cdots,\tilde{a}(\mathbf{x}_k)$ on the vacuum from the right. Hence, the second-quantized position operator given in Equation (\ref{cop}) agrees with the definition of an operator creating the coordinate space, that is needed for the calculation of probability amplitudes of positions of the particles. 
	\\\\
To understand \textbf{where} (spatially) the particles are, the coordinate space displays more importance, than the configuration space. Particularly, it is essential to obtain a deeper  understanding of  QFT's that are defined on non-commutative spaces. 
	\begin{remark}
		An integral transformation from the configuration space to the coordinate space exists and it is given by,
		\begin{align*}\vert x\rangle=
		\phi^{+}(x)\vert0\rangle&= \int d^n\mu(\mathbf{p}) \,e^{ipx}a_c^{*}(\mathbf{p})\vert0\rangle
		\\&= (2\pi)^{-n/2}\int d^n \mathbf{y}\left( \int d^n \mu(\mathbf{p}) \,\sqrt{2\omega_{\mathbf{p}}}\,
		e^{i\left(\omega_{\mathbf{p}}x^0+p_k(x-y)^{k}\right)}   \right) \tilde{a}^*(\mathbf{y})\vert0\rangle,	 
		\end{align*}  
		where in the last line the Fourier-transformation was used. 
	\end{remark}  
	
	\section{Poincar\'e Group} 
	In this section, we take the representations of the Poincar\'e group in terms of QFT Fock-space operators  and perform a base change to the coordinate space.\\\\

	\begin{lemma}The zero component, i.e. the operator generating time-translations, has the following representation in coordinate space,
		\begin{align*}  
		P_{0}=\int d^n \mathbf{x} \,  
		\tilde{a}^*(\mathbf{x})  \left(\tilde{\omega}\ast
		\tilde{a}\right) (\mathbf{x} ),
		\end{align*}
		where  $\ast$ denotes the convolution and the function $\tilde{\omega}(\mathbf{x})$ is defined as 	
	 \[ \tilde{\omega}(\mathbf{x}):=- 2 { {( {2\pi})^{- \frac{n+1}{2} }}}(\frac{m}{|\mathbf{x} |})^{\frac{n+1}{2}}K_{\frac{n+1}{2}}(m|\mathbf{x} |). \] 
	The momentum operator, that is responsible for spatial translations  takes in   coordinate space the following form, 
		\begin{align*} 
		P_j&=
		i
		\int
		d^n \mathbf{x}\,\tilde{a}^*(\mathbf{x})  
		\frac{\partial}{\partial x^{j}}\tilde{a} (\mathbf{x} ) .
		\end{align*}
	\end{lemma}
	\begin{proof}
		Let us first prove the base change for the spatial part of the momentum operator.  The  base change is straight forward and it is done as in the proof for the coordinate operator (see Lemma \ref{NWP}),
		\begin{align*} 
		P_j&=  \int
		d^n \mathbf{p}\,  p_{j}\, {a}^*(\mathbf{p}) {a}(\mathbf{p})\\&=
		(2\pi)^{-n} \int
		d^n \mathbf{p}\,  p_{j}  \int
		d^n \mathbf{x}\, e^{-ip_lx^l} \tilde{a}^*(\mathbf{x}) \int
		d^n \mathbf{y} \,e^{ip_ky^k} \tilde{a} (\mathbf{y} ) \\&=(2\pi)^{-n} 
		\iint
		d^n \mathbf{x}\,d^n \mathbf{y}\, \left( \int
		d^n \mathbf{p}\,  p_{j} \, e^{-ip_k(x-y)^k} \right)\tilde{a}^*(\mathbf{x})  
		\tilde{a} (\mathbf{y} ) \\&=-(2\pi)^{-n} 
		\iint
		d^n \mathbf{x}\,d^n \mathbf{y}\,  i\frac{\partial}{\partial y^{j}}\left( \int
		d^n \mathbf{p}\,  \, e^{-ip_k(x-y)^k} \right)\tilde{a}^*(\mathbf{x})  
		\tilde{a} (\mathbf{y} ) \\&=  i 
		\iint
		d^n \mathbf{x}\,d^n \mathbf{y}\, \delta^{n}(\mathbf{x}-\mathbf{y})\tilde{a}^*(\mathbf{x})  
		\frac{\partial}{\partial y^{j}}\tilde{a} (\mathbf{y} ) \\&
		= i
		\int
		d^n \mathbf{x}\,\tilde{a}^*(\mathbf{x})  
		\frac{\partial}{\partial x^{j}}\tilde{a} (\mathbf{x} ) ,
		\end{align*}
		where in the last lines we performed  a partial integration  and integrated over the delta function. Next, we transform the zero component into the coordinate space. This transformation requires more work and in particular we use Formula (\ref{f2}), \\
		\begin{align*} 
		P_{0} &=  \int
		d^n \mathbf{p}\,  \omega_{\mathbf{p}} \, {a}^*(\mathbf{p}) {a}(\mathbf{p})\\&=(2\pi)^{-n} 
		\int
		d^n \mathbf{p}\,  \sqrt{|\mathbf{p}|^2+m^2} \int
		d^n \mathbf{x}\, e^{-ip_lx^l} \tilde{a}^*(\mathbf{x}) \int
		d^n \mathbf{y} \,e^{ip_ky^k} \tilde{a} (\mathbf{y} ) \\&=(2\pi)^{-n} 
		\iint
		d^n \mathbf{x}\,d^n \mathbf{y}\, \left( \int
		d^n \mathbf{p}\, \sqrt{|\mathbf{p}|^2+m^2} \, e^{-ip_k(x-y)^k} \right)\tilde{a}^*(\mathbf{x})  
		\tilde{a} (\mathbf{y} ) \\&=- 2 { {( {2\pi})^{- \frac{n+1}{2} }}} 
		\iint
		d^n \mathbf{x}\,d^n \mathbf{y}\, 
		(\frac{m}{|\mathbf{x}-\mathbf{y}|})^{\frac{n+1}{2}}K_{\frac{n+1}{2}}(m|\mathbf{x}-\mathbf{y}|)
		\tilde{a}^*(\mathbf{x})  
		\tilde{a} (\mathbf{y} ),
		\end{align*}
		where the Fourier transformation in the last step can be found in   \cite[Chapter III, Section 2.8]{GS1}.
	\end{proof} 
	An interesting object that appeared naturally in \cite{Muc3} is that of the velocity operator. It can be obtained in two different ways. The first path is guided by intuition. Since the operator must be the second quantization of the spatial velocity for a relativistic particle we have the following expression,
	\begin{align*}
	V_{j}&= \int d^n \mathbf{p}\, \frac{p_j}{\omega_{\textbf{p}}}\, {a}^*(\mathbf{p}) {a}(\mathbf{p})
	.
	\end{align*}
	The second more profound approach is given by the Heisenberg equation of motion,  
	\begin{align}\label{heq}
	[P_0,X_j]=-iV_j.
	\end{align}
	The results of the two paths agree  and it was proven in \cite{Muc3}. Since we have expressions for the momentum and  coordinate operator in the spatial space, we can express the velocity operator in the coordinate space by taking the commutator (\ref{heq}).  Hence, the second path is from a calculative point of view easier and the explicit result is given in the following lemma.
	\begin{lemma}
		The velocity operator expressed in terms of ladder operators of the coordinate space is given by,
		\begin{align}
		V_j= -i\int d^n \mathbf{x} \,  
		\tilde{a}^*(\mathbf{x})  \left(\tilde{\omega}_j\ast
		\tilde{a}\right) (\mathbf{x} ),
		\end{align}
		with vector-valued function $\tilde{\omega}_j$ defined as $\tilde{\omega}_j(\mathbf{x}):=   2{ {( {2\pi})^{- \frac{n+1}{2} }}} (\frac{m}{|\mathbf{x} |})^{\frac{n+1}{2}}K_{\frac{n+1}{2}}(m|\mathbf{x} |)\, x_j$ and where $\ast$ denotes the convolution.  
	\end{lemma}
	\begin{proof}
		The proof is straight-forward and it is done by explicitly calculating the Heisenberg equation of motion in coordinate space, i.e. 
		\begin{align*}
		[P_0,X_j]&=\iint d^n \mathbf{x} \, d^n \mathbf{z} \,  z_j\, [
		\tilde{a}^*(\mathbf{x}) \left(\tilde{\omega}\ast
		\tilde{a}\right) (\mathbf{x} ),\tilde{a}^*(\mathbf{z}) \tilde{a} (\mathbf{z}) ]\\&=
		\iiint d^n \mathbf{x} \, d^n \mathbf{y}\,d^n \mathbf{z}\,z_j  \, \tilde{\omega}(x-y)\underbrace{[
			\tilde{a}^*(\mathbf{x})  
			\tilde{a} (\mathbf{y} ),\tilde{a}^*(\mathbf{z}) \tilde{a} (\mathbf{z}) ]}_{-\delta(\mathbf{x}-\mathbf{z}) \tilde{a}^*(\mathbf{z})\tilde{a} (\mathbf{y})+\delta(\mathbf{y}-\mathbf{z}) \tilde{a}^*(\mathbf{x})\tilde{a} (\mathbf{z})}\\&=
		-\iint d^n \mathbf{x} \, d^n \mathbf{y}\,  \,(x-y)_j  \, \tilde{\omega}(x-y)  \tilde{a}^*(\mathbf{x}) \tilde{a}(\mathbf{y}).
		\end{align*}
	\end{proof}$\,$\\
	The velocity operator expressed in the coordinate space is useful for the boost operators. In particular, the boosts of the restricted Lorentz group have the velocity operator explicitly in their respective representations. \\\\
	Next, we turn to the generators of the proper orthochronous Lorentz group $\mathscr{L}^{\uparrow}_{+}$, consisting of boosts  and rotations which are given in the momentum space as, \cite[Equation
	3.54]{IZ}, \cite{SS}
	\begin{align}\label{lbcaop1}
	M_{j0}&= i\int  d^n\mathbf{p}\,{a}^{*}(\textbf{p}) 
	\left(  \frac{p_j}{2\omega_{\textbf{p}}}-\omega_{\textbf{p}}\frac{\partial}{\partial p^j } \right) a(\textbf{p}),
	\\ \label{lbcaop2}
	M_{ik}&=i \int d^n\mathbf{p}\,  {a}^{*}(\textbf{p}) 
	\left(p_i \frac{\partial}{\partial p^k }-p_k\frac{\partial}{\partial p^i }\right)
	a(\textbf{p}).
	\end{align}	The next theorem displays the importance of the NWP-operator. In particular, the generators of   boost and rotations can be represented  by   second quantized (denoted by $d\Gamma(\cdot)$ see \cite[Chapter X.7]{RS2}) products of the momentum and NWP-operator.
	\begin{theorem}For the massive scalar field,  generators of the proper orthochronous Lorentz-group $\mathscr{L}^{\uparrow}_{+}$  can be represented in terms  of products of the NWP-operator and the relativistic momentum operator as follows,
		\begin{align}\label{p1}
		M_{0j}=\frac{1}{2}\big( d\Gamma(X_jP_0)+d\Gamma(P_0X_j)\big)  ,
		\qquad 
		M_{ik} =   d\Gamma(X_iP_k)-d\Gamma(X_{k}P_i).
		\end{align}
	\end{theorem}
	\begin{proof}
		We start the proof for boosts where we have,
		\begin{align*}
		M_{0j}&=\frac{1}{2}\big( d\Gamma(X_jP_0)+d\Gamma(P_0X_j)\big) =\frac{1}{2}  d\Gamma([X_j,P_0])+d\Gamma(P_0X_j) 
		\\  &=\frac{i}{2}  d\Gamma(V_j)+d\Gamma(P_0X_j) ,
		\end{align*}
		where in the last lines we used the Equation of motion given by Formula (\ref{heq}).  The proof for rotations is obvious and hence it is omitted. 
	\end{proof}
The former Theorem conveys the fact that we can define the generator of the Lorentz-group by using the coordinate operator and the momentum operator.  Therefore, the canonical commutation relations with  addition of the zero component of   momentum  can be used to define the group of relativity. This fact, is, in our opinion, an additional  argument for the physical sense of the NWP-operator. \\\\ Moreover,   the representation of   Lorentz generators by using the NWP and momentum operator induces more transparency into
the non-covariant behavior of the   NWP-operator with regards  to boosts. The following calculations clarify the former statement, 
	\begin{align*}
	[M_{0j},P_0]&= \frac{1}{2}[\big( d\Gamma(X_jP_0)+d\Gamma(P_0X_j)\big),d\Gamma (P_0)]\\&=
	\frac{1}{2} \big( d\Gamma([X_j,P_0]P_0)+d\Gamma(P_0[X_j,P_0])\big) 
	\\&=i \, d\Gamma(V_jP_0)=i \, d\Gamma(P_j)=iP_j,
	\end{align*}
	where in the last lines we used  $[d\Gamma(A),d\Gamma(B)]=d\Gamma([A,B])$ 
	and the Heisenberg-equation (see Equation \ref{heq}). 
	Moreover, we used the representation of the velocity operator  as $V_j=d\Gamma(P_0^{-1 }P_j)$. Next,  the commutator of boosts with the spatial momentum is calculated,
	\begin{align*}
	[M_{0j},P_k]&= \frac{1}{2}[\big( d\Gamma(X_jP_0)+d\Gamma(P_0X_j)\big),d\Gamma (P_k)]\\&=
	\frac{1}{2} \big( d\Gamma([X_j,P_k]P_0)+d\Gamma(P_0[X_j,P_k])\big) 
	\\&=-i\eta_{jk}P_0,
	\end{align*}
	where in the last lines we used the second-quantization of the unit operator to be the particle number operator, i.e. $d\Gamma(\mathbb{I}_{\mathscr{H}})=N$ and the explicit canonical commutation relations, see Equation (\ref{ccr}).   Next, we calculate the commutator which  essentially answers the question of covariance w.r.t. the NWP-operator, 
	\begin{align*}
	[M_{0j},X_k]&= \frac{1}{2}[\big( d\Gamma(X_jP_0)+d\Gamma(P_0X_j)\big),d\Gamma (X_k)]\\&=
	\frac{1}{2} \big( d\Gamma(X_j[P_0,X_k] )+d\Gamma( [P_0,X_k]X_j)\big)  \\&=
	\frac{i}{2} \big( d\Gamma(X_jV_k )+d\Gamma( V_kX_j)\big) . 
	\end{align*} The former commutator has no resemblance to  a covariant commutator-like object as $[M_{0j},X_k]= -i\eta_{jk}X_0$. However,  the boost operator can be  translated in time by using the generators of the translation group,
	\begin{align*}
	e^{ix^{0}P_{0}}M_{j0} e^{-ix^{0}P_{0}}&=M_{j0}+ix^{0}[P_{0},M_{j0}]+\frac{i^2}{2!}(x^{0})^2
	\underbrace{ [P_{0},[P_{0},M_{j0}]]}_{=0}+\cdots
	\\&= M_{j0}+x^{0} P_{j}, 
	\end{align*}
with $x^0\in\R$	and where the Backer-Campbell-Hausdorff formula and the explicit Poincar\'e algebra was used. Therefore, the time-dependent boost-operator can be represented as a  symmetric product of the second quantized coordinate, the time and the momentum operators,
	\begin{align*}
	M_{0j}&=\frac{1}{2}\big( d\Gamma(X_jP_0)+d\Gamma(P_0X_j)
	+x_0 d\Gamma(P_0)+ d\Gamma(P_0) x_0
	\big)  \\&=\frac{1}{2}\big( d\Gamma(X_jP_0)+d\Gamma(P_0X_j)	\big)  
	+x_0 d\Gamma(P_0).
	\end{align*}
By using the time-translated boost operator we obtain for the commutator with the NWP-operator ,
	\begin{align*}
	[ M_{j0}+x^{0} P_{j},X_k]&=	\frac{i}{2} \big( d\Gamma(X_jV_k )+d\Gamma( V_kX_j)\big) +x^{0} [ P_{j},X_k]
	\\&= i\eta_{jk}\, x^{0} N+	\frac{i}{2} \big( d\Gamma(X_jV_k )+d\Gamma( V_kX_j)\big),
	\end{align*}
	which has a resemblance to  a covariant operator.  From the view point of representing  the Lorentz generators  by the  use of the NWP- and the momentum operator,  the non-covariance issue becomes clear on a more profound level.  Commutator relations for the rotations are   obvious and hence we shall omit them. However, from the representation of these operators, it becomes  obvious why both the NWP and the momentum operator transform covariantly w.r.t rotations. In particular,  generators of rotations have the well-known complementarity between momentum and position.
	
	\begin{lemma}\label{lg}
		Generators of the  proper orthochronous Lorentz group $\mathscr{L}^{\uparrow}_{+}$ expressed in the terms of ladder operators of the free massive scalar field are represented in the coordinate space as,
		\begin{align*}
		M_{j0}=\frac{1}{2}  	\iint d^n \mathbf{x} \, d^n \mathbf{y}\, (x+y)_j\, \tilde{\omega}(\mathbf{x}-\mathbf{y})  \tilde{a}^*(\mathbf{x}) \tilde{a}(\mathbf{y})  .
		\end{align*}
		The operator of rotations takes  the familiar   form in the coordinate space,
		\begin{align*}
		M_{ik}=	 i \int
		d^n \mathbf{x}  \,
		\tilde{a}^*(\mathbf{x})
		\left(x_i \frac{\partial}{\partial x^k }-x_k\frac{\partial}{\partial x^i }\right) \tilde{a} (\mathbf{x}).
		\end{align*}
	\end{lemma}
	\begin{proof}
		We start by calculating the base change of  the rotation  generators of the Lorentz group  $\mathscr{L}^{\uparrow}_{+}$,
		\begin{align*}
		M_{ik}&=i \int d^n\mathbf{p}\,  {a}^{*}(\textbf{p}) 
		\left(p_i \frac{\partial}{\partial p^k }-p_k\frac{\partial}{\partial p^i }\right)
		a(\textbf{p})\\&=i(2\pi)^{-n}
		\int d^n\mathbf{p}\,  \int
		d^n \mathbf{x}\, e^{-ip_lx^l} \tilde{a}^*(\mathbf{x}) 
		\left(p_i \frac{\partial}{\partial p^k }-p_k\frac{\partial}{\partial p^i }\right)
		\int
		d^n \mathbf{y}\, e^{ip_ry^r} \tilde{a} (\mathbf{y})
		\\&=-i(2\pi)^{-n}
		\int d^n\mathbf{p}\,  \int
		d^n \mathbf{x}\, e^{-ip_lx^l} \tilde{a}^*(\mathbf{x})
		\int
		d^n \mathbf{y}\,\left(y_i \frac{\partial}{\partial y^k }-y_k\frac{\partial}{\partial y^i }\right) e^{ ip_ry^r} \tilde{a} (\mathbf{y})
		\\&=i(2\pi)^{-n}\iint
		d^n \mathbf{x}\, d^n \mathbf{y}\, \left({\int d^n\mathbf{p}\, e^{-ip_l(x-y)^l}} \right)
		\tilde{a}^*(\mathbf{x})
		\left(y_i \frac{\partial}{\partial y^k }-y_k\frac{\partial}{\partial y^i }\right) \tilde{a} (\mathbf{y}) \\&=i \int
		d^n \mathbf{x}  \,
		\tilde{a}^*(\mathbf{x})
		\left(x_i \frac{\partial}{\partial x^k }-x_k\frac{\partial}{\partial x^i }\right) \tilde{a} (\mathbf{x}),
		\end{align*}
		where in the last lines we performed a partial integration and integrated over a delta function. 
		Next, we turn to the Lorentz boosts. The first term is simply the velocity operator times $\frac{i}{2}$. Hence, we focus here only on the second part,
		\begin{align*}
		&- i\int  d^n\mathbf{p}\,{a}^{*}(\textbf{p}) \, \omega_{\textbf{p}}\frac{\partial}{\partial p^j}  a(\textbf{p})\\&=
		- i\int  d^n\mathbf{p}\,\int
		d^n \mathbf{x}\, e^{-ip_rx^r} \,\tilde{a}^*(\mathbf{x}) \, \omega_{\textbf{p}}\frac{\partial}{\partial p^j}  \int
		d^n \mathbf{y}\, e^{ ip_ky^k} \tilde{a}(\mathbf{y}) 
		\\&= (2\pi)^{-n}\iint	d^n \mathbf{x}\,	d^n \mathbf{y}
		\left( \int  d^n\mathbf{p}\, \sqrt{|\mathbf{p}|^2+m^2} \, e^{-ip_k(x-y)^k}    \right)\,y_j\,\tilde{a}^*(\mathbf{x}) \tilde{a}(\mathbf{y}) 
		\\&= -2 { {( {2\pi})^{- \frac{n+1}{2} }}} 
		\iint
		d^n \mathbf{x}\,d^n \mathbf{y}\, y_j\,
		(\frac{m}{|\mathbf{x}-\mathbf{y}|})^{\frac{n+1}{2}}K_{\frac{n+1}{2}}(m|\mathbf{x}-\mathbf{y}|)
		\tilde{a}^*(\mathbf{x})  
		\tilde{a} (\mathbf{y} )
		\\&= 
		\iint
		d^n \mathbf{x}\,d^n \mathbf{y}\, y_j\,\tilde{\omega}(\mathbf{x}-\mathbf{y})
		\tilde{a}^*(\mathbf{x})  
		\tilde{a} (\mathbf{y} ),
		\end{align*}
		where in the last lines we performed the differentiation and used Integral (\ref{f2}).
	\end{proof}
	Now the result for the rotations is not  surprising and is well-known from quantum mechanics, i.e. for one-particle. The interesting and  unknown result is that of the Lorentz-boosts. It is a representation of the boost operator in terms of the spatial coordinate space.  From the form of the boost generator it also becomes more clear why a covariant coordinate operator, i.e. a zero  component (besides issues with Pauli's theorem) $X_{0}$ cannot exist. This can be seen by taking the commutator of the Lorentz boost and the coordinate operator, 
	\begin{align*}
	[M_{j0},X_k]&=	\frac{1}{2}  	\iiint d^n \mathbf{x} \, d^n \mathbf{y}\, d^n \mathbf{z} \,(x+y)_j\,z_k\, \tilde{\omega}(\mathbf{x}-\mathbf{y})  [\tilde{a}^*(\mathbf{x}) \tilde{a}(\mathbf{y}) , \tilde{a}^*(\mathbf{z}) \tilde{a}(\mathbf{z})]\\&=-
	\frac{1}{2}  	\iint d^n \mathbf{x} \, d^n \mathbf{y}\,  \,(x+y)_j\,(x-y)_k\, \tilde{\omega}(\mathbf{x}-\mathbf{y})   \tilde{a}^*(\mathbf{x}) \tilde{a}(\mathbf{y})\\&  \neq
	-i\delta_{jk}X_{0}.
	\end{align*}
	Moreover, the complementarity in an operational sense, i.e. by replacing   multiplication operators with derivatives and vice versa, is broken by the boost operators. This can be easily seen due to their explicit form in the coordinate space. 
	\\\\
	In Section \ref{NWPsec} we gave  the explicit  operator (see Equation (\ref{copes})) to calculate the probability amplitude for  finding $k$-particle at positions $\mathbf{x}_{1} \cdots \mathbf{x}_{k}$ at time $x_{0}={x}_{10}= \cdots{x}_{k0}$. By setting the time component $x_{0}=0$ equal to zero we were able to match this expression with our ladder operator in coordinate space. Hence, we could (or should) have calculated the more general and explicit expressions for the operators by keeping the time component. To resolve this issue, time-translations can be performed by using the unitary adjoint action of the zero component of the momentum. This fact is composed as a result in the form of a proposition.
	\begin{proposition}\label{sf}
		The operator $\phi_1$ given as,
		\begin{align*} 
		\phi_1( {x}) =(2\pi)^{-n/2}\,
		\int d^n \mathbf{p}\, e^{i {p} \cdot {x}}\,a(\mathbf{p}),
		\end{align*}
		can be expressed by using time-translations of the Fourier-transformed annihilation operator $\tilde{a}$
		as follows
		\begin{align*} 
		\phi_1( {x}) =	\phi_1( x^{0}, \mathbf{x})  = e^{ix^{0}P_{0}} \tilde{a}(	 \mathbf{x})e^{-ix^{0}P_{0}}, \qquad x^0\in\mathbb{R}.
		\end{align*}
		
	\end{proposition}
	\begin{proof}
		The proof is performed by using the inverse Fourier-transformation  of $\tilde{a}$, i.e. 
		\begin{align*} 
		e^{ix^{0}P_{0}} \tilde{a}(	 \mathbf{x})e^{-ix^{0}P_{0}}&=(2\pi)^{-n/2} \int
		d^n \mathbf{p}\, e^{-ip_kx^k}e^{ix^{0}P_{0}}{a}(\mathbf{p})e^{-iy^{0}P_{0}}\\&=(2\pi)^{-n/2} \int
		d^n \mathbf{p}\,  e^{-ip_kx^k}\,e^{-ip_{0}x^{0} }{a}(\mathbf{p}) , 
		\end{align*}
		where in the last lines we used the adjoint action of the translation group on the annihilation operator, see \cite[Chapter 7, Equation 61]{Sch}. This agrees with the  expression given in Equation (\ref{copes}). The explicit expression of this operator in  terms of the coordinate space ladder operators is given in the next section.
	\end{proof} Hence, instead of using the time-independent transformations using the ladder operators in coordinate space, we could have used the product of $\phi_1(x^0,\mathbf{x})$ and $\phi_1(x^0,\mathbf{y})$ \textbf{or} we can simply take our obtained expressions and perform a time-translation. However, from the Baker-Campbell-Hausdorff formula and the particular algebra of the Poincar\'e group it follows that  the momentum operator and the spatial rotations are invariant under time-translations. Hence, the only term that changes under time-translation is the boost operator and this is  easily calculated,
	\begin{align*}
	e^{ix^{0}P_{0}}M_{j0} e^{-ix^{0}P_{0}}  
 = M_{j0}+x^{0} P_{j},
	\end{align*}
using the Baker-Campbell-Hausdorff formula and the    Poincar\'e algebra.	
	\\\\
	Although, neither the coordinate operator nor the respective eigenstates are Lorentz-covariant, the Poincar\'e operators that are translated in this work obey the covariant transformation property. In particular, the Lorentz-covariance does not depend on the  specific base that is chosen. This statement is proven  by re-transforming the ladder operators from coordinate space into momentum space, and taking the appropriate transformation into account. This is obvious and it simply follows from the fact that that the operators are covariant, invariant of the representation they are given in. Hence, even though we use non-covariant eigenstates, generators of the Poincar\'e group written in coordinate space  respect the relativistic covariance.
	\section{Poincar\'e Transformations of the Coordinate Space}
	In this section we investigate   Poincar\'e transformations, for the massive case,  of the Fourier-transformed ladder operators. In particular, we calculate  how the operator $\phi_1$, which is the object used to calculate probability amplitudes, transforms under space-time translations and pure rotations. The motivation herein is to give the proper transformation behavior of  probability amplitudes under changes of frame. In particular, given two frames related by space-time translations or pure rotations, both must be able to calculate and compare the probability amplitude w.r.t. the frame change. We purposely neglect boosts since the spatial coordinate space is not Lorentz-covariant and hence the question of how boosts act is not well-defined from the very beginning.  \\\\In the forthcoming calculations, we	work in the physical relevant dimension $d=4$. Moreover, we define the unitary operator that generates transformations  of the orthochronous proper Poincar\'e group, $\mathscr{P}^{\uparrow}_{+}=\mathscr{L}^{\uparrow}_{+}\ltimes\mathbb{R}^4$ as $U(y,\Lambda)$ transforming the creation and annihilation operators as, \cite[Chapter 7]{Sch},
	\begin{align}\label{traf}
	U(y,\Lambda) a(\mathbf{p})U(y,\Lambda)^{-1}&= \sqrt{\frac{	\omega_{\Lambda \mathbf{p}}}{	\omega_{ \mathbf{p}}}}e^{-i(\Lambda p)_{\mu}y^{\mu}}{a}(\Lambda\mathbf{p}) ,\qquad (y,\Lambda) \in \mathscr{P}^{\uparrow}_{+},\\\label{traf1}
	U(y,\Lambda) a^{*}(\mathbf{p})U(y,\Lambda)^{-1}&= \sqrt{\frac{	\omega_{\Lambda \mathbf{p}}}{	\omega_{ \mathbf{p}}}}e^{ i(\Lambda p)_{\mu}y^{\mu}}{a}^{*}(\Lambda\mathbf{p}) ,\qquad (y,\Lambda) \in \mathscr{P}^{\uparrow}_{+}
	.
	\end{align} 
	To simplify calculations with regard  to transformations of the operator $\phi_1$, we first give the adjoint actions on the Fourier-transformed operators and afterwards perform a time-translation, see 
	Proposition \ref{sf}. 		By using the former results, we calculate transformations in the momentum space and perform an inverse Fourier-transformation in order to get the respective coordinate space transformations.

	\begin{lemma}\label{l61}
		The Fourier-transformed annihilation  operator $\tilde{a}$  transforms under the adjoint action of the translation group as follows,    
		\begin{align*}
		U(y ,\mathbb{I}) \tilde{a} (	 \mathbf{x}) U( y ,\mathbb{I}) ^{-1}& =\frac{iy^{0} }{2\pi^2}  \left(
		\frac{	 m^2K_{2}(m\sqrt{\vert \mathbf{x+y} \vert^{2}-(y^{0})^2})}{\vert \mathbf{x+y} \vert^{2}-(y^{0})^2}\right)\ast\tilde{a}(\mathbf{x+y}) ,
		\end{align*}
		where $y\in\mathbb{R}^4$ and $\ast$ denotes the convolution.
	\end{lemma}
	\begin{proof}
		Since, the momentum operators commute we can first calculate the spatial transformations, which are easier, and then the time-translations, i.e. 
		\begin{align*} 
		U( {y},\mathbb{I})  \tilde{a}(	 \mathbf{x})  U( {y},\mathbb{I}) ^{-1} =
		U(y^{0},\mathbb{I})\,
		U(\mathbf{y},\mathbb{I})  \,\tilde{a}(	 \mathbf{x}) \, U(\mathbf{y},\mathbb{I}) ^{-1}  U(y^{0},\mathbb{I})^{-1}.
		\end{align*}
		Therefore, 	  we first investigate the following expression,
		\begin{align*} U(\mathbf{y},\mathbb{I})  \tilde{a}(	 \mathbf{x})  U(\mathbf{y},\mathbb{I}) ^{-1}&= e^{iy^kP_{k}} \tilde{a}(	 \mathbf{x})e^{-iy^kP_{k}}
		\\&=
		(2\pi)^{3/2} \int
		d^{3} \mathbf{p}\, e^{-ip_kx^k}    U(\mathbf{y},\mathbb{I}) {a}(\mathbf{p})  U(\mathbf{y},\mathbb{I}) ^{-1}\\
		&=(2\pi)^{-3/2} \int
		d^{3} \mathbf{p}\, e^{-ip_k(x+y)^k}     {a}(\mathbf{p})= \tilde{a}(	 \mathbf{x}+\mathbf{y}),
		\end{align*}
		where in the last lines we used the inverse Fourier-transformation (see Equation \ref{inf}) and the well-known action of translations in the momentum space, see \cite[Chapter 7, Equation 61]{Sch}. 
		Next, we calculate the action of time-translations, which is
		given by the following expression,
		\begin{align*} U(y^{0},\mathbb{I}) \tilde{a}(	 \mathbf{x})  U(y^{0},\mathbb{I})^{-1} &=e^{iy^{0}P_{0}} \tilde{a}(	 \mathbf{x})e^{-iy^{0}P_{0}}\\
		&=(2\pi)^{-3/2} \int
		d^{3} \mathbf{p}\, e^{-ip_kx^k}   e^{-iy^{0}  \omega_{\mathbf{p}}}   {a}(\mathbf{p})  \\
		&=(2\pi)^{-3 }\int
		d^{3} \mathbf{z}    {\left(\int
			d^{3} \mathbf{p}\, e^{-ip_k(x-z)^k} e^{-i  \omega_{\mathbf{p}}y^{0}}\right)}_{
		}  \tilde{a}(\mathbf{z})\\&=\frac{i	y^{0}}{2\pi^2}
		\int
		d^{3} \mathbf{z} \left(
		\frac{ m^2K_{2}(m\sqrt{\vert \mathbf{x}-\mathbf{z}\vert^{2}-(y^{0})^2})}{\vert \mathbf{x}-\mathbf{z}\vert^{2}-(y^{0})^2}\right)\tilde{a}(\mathbf{z}),
		\end{align*}
		where in the former equation we used \cite[Chapter 2]{PS} and \cite[Equation 3.914]{RG} to solve the integral. Equivalent considerations can be done for the Fourier transformed creation operator $\tilde{a}^*$.	 By using both transformations the proof is completed. 
	\end{proof}
	The translation acting on the coordinate space in spatial direction was as expected a translation in the spatial space as well. However, the difference in the spatial space becomes visible when we perform a translation  in time. Since, there is no time component the time-translation transforms, and therefore acts, on  the spatial space in a more complex manner.  It is interesting since this is a physical   important expression. Essentially, it tells us how the spatial coordinate space (which is made of $k$-ladder operators, see Equation \ref{copes1}), transforms under a time-translation. Moreover, with regards to the time-translated coordinate annihilation operator $\tilde{a}$, which corresponds to $\phi_1$, an explicit expression in coordinate space was given. The next object of interest is the Lorentz-transformation of   the Fourier-transformed ladder operators. 
	\begin{lemma}\label{l62}
		Let $\Lambda_R$ denote the matrices of the Lorentz group which represent pure rotations. They are given by 
		\[\Lambda_R=\left(
		\begin{matrix} 
		1 & 0\\
		0 & R	 
		\end{matrix}\right), \qquad R\in SO(3),
		\]
		where $SO(3)$ is the group of real, orthogonal, $3\times 3$ matrices  with determinant one. Then, the Fourier-transformed ladder  operator  $\tilde{a}$  transforms under pure rotations as follows  
		\begin{align*}
		U(0,\Lambda_R)  \tilde{a}(	 \mathbf{x}) U(0,\Lambda_R) ^{-1} =
		\tilde{a} (	R \mathbf{ x }).
		\end{align*} 
	\end{lemma}
	\begin{proof}
		The action of pure spatial rotations  on the  the Fourier-transformed annihilation operator is calculated by using Equation (\ref{traf1}), 
		\begin{align*}
		U(0,\Lambda_R)  \tilde{a} (	 \mathbf{x})  	U(0,\Lambda_R)  ^{-1} &=
		(2\pi)^{-3/2} \int
		d^{3} \mathbf{p}\, e^{-ip_kx^k}   U(0,\Lambda_R){a}(\mathbf{p}) U(0,\Lambda_R)^{-1}\\&=
		(2\pi)^{-3/2} \int
		d^{3} \mathbf{p}\, e^{-ip_kx^k}   \sqrt{\frac{	\omega_{	R\mathbf{p}}}{	\omega_{ \mathbf{p}}}}{a}(	R\mathbf{p})  
		\\&=
		(2\pi)^{-3/2} \int
		d^{3} \mathbf{p}\, e^{-i p_k(	R x)^k}    {a}( \mathbf{p}),
		\end{align*}
		in the last lines we used the transformation behavior of the non-covariant momentum ladder operators  (see \cite[Chapter 7, Equation 57]{Sch}) and the orthogonality of $R$ for pure spatial rotations.  
		The proof for the Fourier-transformed creation operator $\tilde{a}^{*}$ can be done equivalently as for the annihilation operator. However, the transformations for $\tilde{a}^{*}$ can be as well obtained by taking the adjoint of the Fourier-transformed annihilation operator. 
	\end{proof}  
	Next, by using the former lemmas we  calculate the transformational behavior of $\phi_1(x)$ under the whole  group of translations and pure rotations. 
	\begin{theorem}
		The  operator	  $\phi_1(x)$ transforms in a covariant manner under space-\textbf{time} translations and pure rotations, i.e. 
			\begin{align*}&
			U(y,\Lambda_R) \phi_1(x)U(y,\Lambda_R)^{-1}
			= \phi_1(x_0+y_0,	R \mathbf{  x+y})   , \qquad (y,\Lambda_{R}) \in \mathscr{P}^{\uparrow}_{+}.
			\end{align*} 
		 Moreover, the explicit transformation   of the Fourier-transformed operator $\tilde{a}$ under the action of the subgroup $(y,\Lambda_{R}) \in \mathscr{P}^{\uparrow}_{+}$ is in coordinate space given as, 
		\begin{align*}&
		U(y,\Lambda_R)\tilde{a}(\mathbf{x})U(y,\Lambda_R)^{-1}
		= 
		\frac{iy^{0} }{2\pi^2}  \left(
		\frac{	 m^2 K_{2}(m\sqrt{\vert \mathbf{x+y} \vert^{2}-(y^{0})^2})}{\vert \mathbf{x+y} \vert^{2}-(y^{0})^2}\right)\ast\tilde{a}(\mathbf{x+y}).
		\end{align*}
	\end{theorem}
	\begin{proof}
		By using Proposition \ref{sf} we can write the operator $\phi_1(x)$ in terms of the Fourier-transformed operator, i.e. 
		\begin{align*}
		U(y,\Lambda_R) \phi_1(x)U(y,\Lambda_R)^{-1}  =
		U(y,\Lambda_R) U(x^0,\mathbb{I})  \tilde{a} (	 \mathbf{x})  U(x^0,\mathbb{I})^{-1}U(y,\Lambda_R)^{-1},
		\end{align*}
		since time-translations commute with space-time translations and pure rotations and since we have $U(y,\Lambda_R)=	U(y,\mathbb{I})	U(0,\Lambda_R)$,  they allow us to  rewrite the former expression as 
		\begin{align*}&
		U(x^0,\mathbb{I}) \,
		U(y,\Lambda_R) \,\tilde{a} (	 \mathbf{x}) \, U(y,\Lambda_R)^{-1}U(x^0,\mathbb{I})^{-1}\\ &=  U(x^0,\mathbb{I}) \,
		U(y,\mathbb{I})	\,U(0,\Lambda_R)\, \tilde{a} (	 \mathbf{x}) \,U(0,\Lambda_R)^{-1} U(y,\mathbb{I})^{-1}	U(x^0,\mathbb{I})^{-1}\\ &=  U(x^0,\mathbb{I})\, 
		U(y,\mathbb{I}) \,\tilde{a} (R \mathbf{ x}) \, U(y,\mathbb{I})^{-1}	U(x^0,\mathbb{I})^{-1}
		\\ &=    U(x^0+y^0,\mathbb{I})\, 
		\,\tilde{a} (R\mathbf{	  x+y}) \, 	U(x^0+y^0,\mathbb{I})^{-1}
		\\ &=   \phi_1(x_0+y_0, R\mathbf{ x+y})   ,
		\end{align*}
		where in the last lines we used Proposition \ref{sf} and the Lemmas \ref{l61} and \ref{l62}, which give the transformational behavior of  the coordinate space operator $\tilde{a}$.
	\end{proof}
	
	\section{Conclusion and Outlook}
	In this paper we performed a base change of the  Poincar\'e group for a massive scalar field, into the coordinate space. This was done by  expressing the creation and annihilation operators, via Fourier-transformation, in terms of the coordinate basis.
	\\\\
	One particular interesting, but expected, aspect of our results was the representation of the spatial momentum and the Lorentzian infinitesimal generators of rotations. The one-particle spatial momentum operator and the rotation operators had the same form as in standard quantum mechanics. This is due to the fact that the energy condition, i.e. the choice of $\omega_{\mathbf{p}}$, does not explicitly enter the expressions for these particular observables.  Nevertheless, for  time-translations and boosts  we notice the difference and in particular the complementarity between the $\mathbf{x}$-space and the momentum space is broken.  
	\\\\
	Concerning the specific actions of the Poincar\'e group on this space we obtained   important results. The action of the spatial translations and rotations on the coordinate space were as expected. However, the interesting result occurred when the time-translation was involved. In a sense, these transformations are a way to fully  understand what happens in the spatial coordinate space,  when transformations involving time come into play. The restriction to this particular space is a by-product of the mass shell condition  that restricted a $d$-dimensional momentum space to a $(d-1)$-dimensional subspace. Hence, by using the transformations in  coordinate space one can  calculate the probability of finding  particles in certain  spatial-positions   at time $x_{0}$ after an explicit   time-translation was performed.
	\\\\ 
	In this work, we translated the boost operators into the coordinate space. Moreover, the adjoint action of the Fourier-transformed operators w.r.t. the boosts was not calculated. Generally, such a calculation is  not well-posed since   boosts mix time and space in   a non-trivial manner. And as hence, the eigenstates of the Newton-Wigner operator are not Lorentz-covariant the calculation of a boost on these particular spaces is from a physical point of view considered problematic.
	\\\\
	We performed all base changes by using  the massive scalar field. However, a possible extension to this framework can be done with regard  to other fields, as for example considering the conformal group expressed in terms of the massless scalar field. This is work in progress. 
	
	\section*{Aknowledgments}
	The author would like to thank Prof. K. Sibold for initiating   deep conceptual questions of this work.  Moreover, we would like to thank Prof. K. Sibold, Prof. Yuri Bonder and S. Pottel for many fruitful discussions during different stages of this paper. The linguistical  corrections of Dr. Z. Much are thankfully acknowledged.

	\bibliographystyle{alpha}
	\bibliography{allliterature1}

\end{document}